\documentclass[a4paper,11pt]{article}
\usepackage{relsize}
\usepackage{epigraph}
\usepackage{nicefrac}
\usepackage[margin=1in]{geometry}
\usepackage{amsmath,amssymb,amsthm}
\usepackage{hyperref}
\usepackage{fullpage}

\hypersetup{colorlinks=true,linkcolor=blue,citecolor=blue}

\usepackage{comment}
\usepackage{subfigure}
\usepackage{algpseudocode}
\usepackage{xcolor}

\usepackage{amsmath, amsthm, amssymb, amstext, comment, graphicx, fullpage}
\usepackage{mathtools,extarrows}

\usepackage{latexsym}

\newcommand{\set}[1]{\ensuremath{\left\{ {#1} \right\}}}
\newcommand{\abs}[1]{\ensuremath{\left| {#1} \right|}}
\newcommand{\exval}[2][]{\ensuremath{\mathop{\mathrm{E}}_{#1} \left[ {#2} \right]}}
\newcommand{\zo}{\ensuremath{\set{0,1}}}

\newcommand{\union}{\ensuremath{\cup}}
\newcommand{\Intersect}{\ensuremath{\bigcap}}
\newcommand{\intersect}{\ensuremath{\cap}}
\newcommand{\st}{\ensuremath{\;\big |\;}}
\newcommand{\bs}{\ensuremath{\backslash}}
\definecolor{coolblack}{rgb}{0.0, 0.18, 0.39}

\DeclareMathOperator*{\argmax}{\arg\!\max}

\newcommand{\type}{\ensuremath{B}}
\newcommand{\strategy}{\ensuremath{\mathcal{X}}}

\newcommand{\result}{\ensuremath{Y}}

\newcommand{\envy}{\ensuremath{Envy}}

\usepackage{xcolor}

\usepackage[ruled,linesnumbered]{algorithm2e} % For algorithms

\SetAlFnt{\small}
\SetAlCapFnt{\small}
\SetAlCapNameFnt{\small}
\SetAlCapHSkip{0pt}
\IncMargin{-\parindent}
\usepackage{url}
\newtheorem{theorem}{Theorem}[section]
\newtheorem{definition}{Definition}
\newtheorem{open}{Open Problem}
\newtheorem{proposition}{Proposition}
\newtheorem{corollary}{Corollary}

\newtheorem{lemma}{Lemma}

\newtheorem{remark}{Remark}
 \allowdisplaybreaks
 
\usepackage{hyperref}
\hypersetup{colorlinks=true,linkcolor=black,citecolor=blue}

\definecolor{blue(pigment)}{rgb}{0.2, 0.2, 0.6}

\title{Sharing information with competitors}
\author{
	Simina Br\^anzei\footnote{Purdue University, USA. E-mail: \href{mailto:simina.branzei@gmail.com}{simina.branzei@gmail.com}.}
	\and
	Claudio Orlandi\footnote{Aarhus University, Denmark. E-mail: 
		\href{mailto:orlandi@cs.au.dk}{orlandi@cs.au.dk}.}
	\and
	Guang Yang\footnote{Intitute of Computing Technology, Chinese Academy of Sciences, China. E-mail: \href{mailto:guang.research@gmail.com}{guang.research@gmail.com}.}
}

\date{}

\begin{document}
	\maketitle
% note that the abstract must come before \maketitle
\begin{abstract}
We study the mechanism design problem in the setting where agents are rewarded using information only. This problem is motivated by the increasing interest in secure multiparty computation techniques. 
More specifically, we consider the setting of a joint computation where 
different agents have inputs of different quality and each agent is interested in learning as much as possible while maintaining exclusivity for information.

Our high level question is to
design mechanisms that motivate all agents (even those with high-quality input) to participate in the computation and we formally study problems such as set union, intersection, and average.
\end{abstract}

\newpage
% note: this command has been disabled to remove the ACM copyright block. Sorry...

% Renew this after \maketitle if the default list of authors is too long for headers
%\renewcommand{\shortauthors}{W.\ Vickrey et.\ al.}
% !TEX root = paper.tex
%%%%%%%%%%%%%%%%%%%%%
%%%	NEW SECTION
%%
%
\section{Introduction}
\label{sec:intro}

Secure multiparty computation allows a set of parties to compute any functions on their private inputs. In recent years there has been a boom
in the speed achieved by cryptographic protocols for secure multiparty computation (see 
e.g.,~\cite{
DBLP:conf/crypto/NielsenNOB12,
DBLP:conf/crypto/DamgardPSZ12,
DBLP:conf/crypto/LindellPSY15,
DBLP:conf/ccs/KellerOS16,
DBLP:conf/asiacrypt/Ben-EfraimLO17,
DBLP:journals/tissec/PinkasSZ18} 
and references therein), to the point that start-ups and companies are beginning to offer products based on these technologies~\cite{DBLP:journals/iacr/ArcherBLKNPSW18}.
One question that has not been addressed in the cryptographic community so far is whether parties will have any \emph{incentive} in participating
in such protocols: In traditional multiparty computation tasks, multiple agents wish to evaluate some public function on their private inputs, where all agents are equal and the evaluated result is broadcasted to all of them or at least the honest ones.
However, when viewing through the game-theoretic lens, the function evaluation process can be realized as the exchanging of private information among those agents,
and hence the agents are \emph{not equal}.
For example, an agent with higher influence on the function tends to have a smaller incentive in the cooperation, and in the extreme case a ``dictator'' would have zero incentive;
or even if the function is symmetric, an agent may still be less incentivized because of a high quality private input which provides a better prior than others. An example of a dictator is a player with input zero when the function is \emph{AND}; such a player already knows the output of the computation and can learn nothing from others.

To this end we suggest to consider the procedure fairness (rather than the result fairness) in terms of \emph{information benefit}, which measures how much an agent improves the quality of her own private information by participating.
We believe this is a better characterization of the agent incentives. 
Also from the game-theoretic point of view, it makes sense to consider the agents as rational and self-motivated individuals rather than simply ``good/bad" or ``honest/semi-honest/malicious" as is typically done in cryptographic scenarios.

In this work, we study mechanisms for exchanging information without monetary transfer among rational agents.
These agents are rational and self-motivated in the sense that they only care about maximizing their own utility defined in terms of information.
More specifically, we focus on utility functions that capture the following properties about the behavior of the agents:

\begin{itemize}
	\item \emph{Correctness}: The agents wish to collect information from other agents.
	
	\item \emph{Exclusivity}: The agents wish to have exclusive access to information.
\end{itemize}

The wish to collect information incentivizes cooperation, while the wish for exclusivity deters it.
By unifying the above competing factors, agents aim to strike a balance between the two. The value of exclusivity is a well known concept studied in many areas of economics (e.g. labor economics, economics of the family, etc); see, e.g. ~\cite{SW00} for a study on the role of exclusivity in contracts between buyers and sellers and ~\cite{MSV10} for platform-based information goods.

Utility functions that capture these competing factors are relevant in modeling situations where both cooperation and competition exist simultaneously, 
such as several companies wishing to exchange their private but probably overlapping information, e.g.~training data for machine learning purpose, predictions for the stock market, etc.
Another related example is where several Twitter marketing agents, with distinct but overlapping sets of followers, collaborate on improving influence by re-tweeting each others' tweets.
Here the influence on followers has very similar properties (i.e. easy replicable but non-fungible) as private information except that it is publicly observable now,
and the cooperation is stable only when everyone gets a fair share from their participation. 

We investigate specific information exchanging problems, such as Multiparty Set Union, as well as Set Intersection and Average.
For example, in the set union problem there is a number of players, each owning a set, and the goal is to find the union of the private sets held by all the agents. Set intersection is similarly defined except the goal is to find the intersection.
Since for such problems the result is not Boolean and agents with different quality input should get different results,
the value of result is measured by quality (accuracy) rather than by a Boolean indicator of whether it is the optimal one. 

For the behavior of the agents, many of our results are for the ``all-or-nothing model'' where every agent either fully participates by truthfully submitting their input or refuses to participate. We also have several results for games with few players in the richer model where agents can partially participate, by submitting some but not all of their information, as well as open questions.
%i.e.~no strategic lie or partial participation is allowed. 
The all-or-nothing model is implemented in practice when the inputs are authenticated by some trusted authority (for instance by mean of digital signatures).
or at least easy to verify afterward so that dishonest agents can be detected and punished eventually (such as with the court or future rounds of repeated games).
% the digests of the actual inputs are either revealed later on or by side channels so that the dishonest agent can eventually be detected and severely punished.}
For example, any deviating behavior is transparent in the ``tweet for tweet'' example.
The participants send their private input to the trusted mediator\footnote{In the secure multiparty computation setting this trusted party is usually replaced by a cryptographic protocol. For the sake of simplicity, we do not further consider cryptographic protocols in this work.} (i.e. ``principal") 
who runs a publicly known protocol (mechanism) to decide the payoff of each agent. 
Here the payoffs are customized pieces of information 
since we are studying the information exchanging mechanism without money.

\begin{comment}
SIMINA's note: I think this could be omitted from here, it's not clear what it means at this point}
Furthermore, we study informational mechanism design for general problems in the \emph{single-envy} case.
Here by ``single-envy"  we mean that every agent cares about the gap between their own information benefit and that of the agent with the highest payoff among others,
which intuitively captures the notion of ``informational advantage over other agents".
\end{comment}

\subsection{Our Contribution}

In this paper we propose a framework for non-monetary mechanism design with utility functions unifying both preferences of correctness and exclusivity.
Let $N = \{1, \ldots, n\}$ be a set of players (agents). Suppose each player has some piece of information, the details of which we intentionally leave informal for now. 
Given some mechanism $M$ that the players use to exchange information among themselves, we define the information benefit $v_i$ of a player $i$ to be the additional ``information'' gained by $i$ after participating in $M$. For example, in the case of the set union problem, where each agent owns a set of elements and tries to learn additional elements from other agents, this gain could be the number of additional elements learned by a player compared to what that player already knew. 

\medskip

The utility function will capture the tension between the wish to learn and the wish for exclusivity and the simple instantiation that we focus on is given by $u_i = v_i - \max_{j \in N \setminus \{i\}} v_j$. Thus each player wishes to learn as much as possible while maintaining exclusivity over the information, which is captured by minimizing the amount obtained by others. Note this definition is connected to the notion of envy-freeness; in particular, it captures the maximum ``envy'' that a player $i$ could have towards any another player, and the goal is to reduce envy.
\medskip

Our technical contribution is to design mechanisms for natural joint computation tasks such as 
\emph{Multiparty Set Union}, as well as \emph{Set Intersection} and \emph{Average}. We focus on mechanisms that incentivize players to submit the information they have as well as ensure properties such as Pareto efficiency\footnote{Pareto effiency ensures that no agent can be better off without making someone else worse off.} of the final allocation. 

\medskip

In the Multiparty Set Union Problem each player owns a set $x_i$ drawn from some universe $\mathcal{U}$. The utility functions are as described above. The strategy space of a player consists of sets they can submit to the mechanism. We assume that agents can hide elements of their set, but not manufacture elements they don't have (i.e. there is a way to detect forgery). The question is to design a mechanism that incentivizes the players to show their set of elements to others.

\begin{theorem}
There is a truthful and Pareto efficient mechanism for set union among $n=3$ players. The mechanism runs in polynomial time.
\end{theorem}

We leave open the general mechanism design question for any number of players.

\begin{open}
	Is there a truthful polynomial time mechanism for set union for any number of players? Are there randomized such mechanisms?
\end{open}

However, we manage to solve this problem for the special case where each player can either submit its whole set or the empty set, i.e. cooperate or not. We call this the ``all-or-nothing'' model. 

\begin{theorem} %\label{thm:multiparty}
	There is a truthful, Pareto efficient, and welfare maximizing mechanism% 
	\footnote{Welfare maximization is achieved by maximizing over all the Pareto efficient outcomes.} for set union among any number $n$ of all-or-nothing players. The mechanism runs in polynomial time for any fixed $n$. 
%	Mechanism~\ref{mech: npsu} guarantees stable cooperation and its output is the Pareto optimal solution with maximal social welfare for Multiparty Set Union problem in the all-or-nothing model.
%This mechanism is also symmetric and strongly dominant.
\end{theorem}

We further show that this mechanism satisfies several other desirable properties, such as treating equal agents equally and rewarding more agents that contribute more.

\medskip

Beyond multiparty set union, we also consider two case studies of problems with sets. The first is a set intersection problem, where each player owns a connected set (interval) on the real line. The players have to find an element in the intersection of all the sets and are promised that such an element exists. A high level example of this problem is when the agents are trying to find a gold mine, and each agent has an estimate of the location of the mine (with a radius around the correct location). The goal of the agents is to improve their estimate of the gold mine, while not revealing the location to too many players if possible.

\begin{theorem}
	There is a truthful polynomial time mechanism for interval intersection among any number $n$ of all-or-nothing players.
\end{theorem}

The second case study is a point average problem, where each player has a point and the goal is to compute the average value of their inputs.

\begin{theorem}
There is a truthful polynomial time mechanism for the point average problem for any number $n$ of all-or-nothing players.
\end{theorem}

We also provide a more general theorem for arbitrary value functions (e.g. that do not necessarily count the number of elements in a set); for this see Appendix~\ref{app:general}.

\medskip

Finally, two more high-level remarks are in order.

\medskip

\noindent \emph{Why not maximize social welfare?} A trivial solution to problems such as set union can be to have everyone learn everything (i.e. maximize the sum of information gains). 
In traditional settings such as auctions or elections %the output space is very limited so that
it is unlikely that every agent maximizes their information benefit simultaneously since their ideal outputs are usually conflicting, 
e.g.~there is only one indivisible good that cannot be assigned to more than one agent.
However, in the world where information replaces material goods, it becomes possible to duplicate the information at (nearly) zero cost such that every agent gets all information and hence maximizes their utility at the same time.
This straightforward mechanism only works if all agents are selfless and choose to report truthfully. 
However, it is unfair in the sense that the more one agent contributes, the less benefit they could get (since the information benefit is bounded by the whole information minus their private information).
Furthermore, the straightforward mechanism fails badly when agents take exclusivity into consideration: 
e.g. the dominant strategy would be ``revealing nothing to the mechanism but combining the output with the private input afterward" and eventually the equilibrium becomes that no exchange happens at all (similar to the  ``rational secret sharing" problem discussed in \cite{HT04,IML05,KN08}) when partial participation and strategic lies are allowed;
and even in the all-or-nothing model an agent may prefer not participating according to their own utility function if their advantage over other competing agents would decrease.

\medskip

\noindent \emph{A Note on Mechanism Design.} The intuition behind our constructions is that every single agent, when participating in the cooperation, should get a benefit no less than the loss they could cause to others by not participating.
At a first glance it might seem that the ``loss to others'' inflicted by a non-participating agent would be bounded by the exclusive information of that agent.
However, it turns out that agents contribute much more to the mechanism than simply their private inputs. In particular, the participation of an agent may increase social welfare by giving incentives for participation to other agents with ``better'' inputs. Concretely, an agent $i$ with a high quality input might choose to join the computation, or reveal more of his private information, because, by doing so, they can reduce the information benefit of some other agent $j$ (which is rational when it reduces $i$'s own exclusivity loss). 

Therefore, the key idea behind our constructions is to characterize the marginal contribution of every agent and assign information accordingly so that nobody prefers to leave the cooperation (and in the meanwhile we aim to maximize the social welfare among all stable allocations).
For example, this idea is instantiated as a round-by-round exchange mechanism for the Three-Party Set Union problem (in Section~\ref{sec:union3}), such that in every single ``round'' of exchange each agent gains more benefit than he offers to others.
%Similar idea is used in the construction for the Multiparty Set Union problem,
%though not as explicit as in the three-party case.

\subsection{Related Work}

%cooperative game and shapley value
Our setting is reminiscent of cooperative game theory and the well-known solution of Shapley value \cite{shapley,roth1988shapley,aumann1989game},
except that now the agents are rewarded with information instead of money.
There are two main distinctions:
a) Information can be duplicated, for free or with negligible cost; 
b) Every piece of information is unique whereas money is fungible, e.g.~the same piece of information could have different values for different agents. 
The first property results in an unfixed total profit (sum of all agents' payoffs) and so breaks the intuition of ``distribute the total surplus \emph{proportionally} to each agent's contribution" used in Shapley value.
The second property requires the mechanism to specify not only the \emph{amount} of information but also the \emph{details} of information allocation.
In particular, the information already contained in an agent's input cannot be used to reward that agent.
Such a property also leads to a subtle dilemma --- 
the more an agent contributes, the less they can get as a reward from the mechanism  ---
e.g.~an agent with all information cannot get new information from other agents.
Therefore, the mechanism must be able to motivate the most informed agents even though they may not benefit as much as those that know less (i.e. with lower quality inputs). A different line of work has studied the problem of sharing information when the inputs are substitutes or complements~\cite{CW16}, which defined the value of information (and of a marginal unit of information) and instantiated it in the context of prediction markets.

%NCC line of work

Our model can be seen as an extension of the \emph{non-cooperative computation} (NCC) framework 
and \emph{informational mechanism design} (IMD) introduced %by Shoham and Tennenholtz 
in \cite{ST05,MPS03}, 
where they characterize Boolean functions that are computable by rational agents with non-monetary utility functions defined in terms of information. 
In their model, the agents are trying to compute a public Boolean function on their private inputs with the help of a trusted center. Every agent claims their type (truthfully or not) to the center, and gets a response from the center
(typically but not necessarily the Boolean function evaluated on claimed types).
Agents may lie or refuse to participate, and they can apply any \emph{interpretation function} (on the response from the center and their true input, so as to correct a wrong answer possibly caused by an earlier false declaration).
In the setting of \cite{ST05}, the agents have a two-tiered preference of correctness preceding exclusivity%
\footnote{\cite{MPS03} considers two more facets, i.e.~privacy and voyeurism, but still in lexicographic ordering.},
i.e.~they are interested in misleading others only if this would not hurt their own correctness,
whereas we generalize this lexicographic preference to a utility function incorporating both components
(The lexicographic preference is a special case when one component is assigned a very small weight).  
Another extension is that we consider non-Boolean functions in this work and allow distinct responses to different agents, which significantly enriches the space of candidate mechanisms. 

%SMPC protocols of SS

The line of work \cite{HT04,IML05,KN08} focuses on the cryptographic implementation of truthful mechanisms for secret sharing and multiparty computation by rational agents without a trusted mediator. 
In their setting there is an ``issuer" who authenticates the initial shares of all agents so that the agents cannot forge a share (just as in the all-or-nothing model).
Then the agents use simultaneous broadcast channels (non-simultaneous channels are also considered in \cite{KN08}) to communicate in a round-by-round manner.
Since all messages are broadcasted in this setting, a rational agent  tends to keep silent so that they can receive others' information without revealing their own and hence possibly gain advantage in exclusivity.
Therefore, much of the efforts and technical depth along this line is spent on catching dishonest agents (who do not broadcast their shares when they are supposed to),
based on the key idea that in any given round the agents do not know whether this is just a test round designed to detect cheaters, or whether it is the final round for the actual information exchange.
\cite{IML05} achieve a fair, rational secure multiparty computation protocol which prevents coalitions and eliminates subliminal channels, 
despite the drawback of requiring special purpose hardware such as ideal envelopes and ballot boxes.  
However,  all of the these works assume the two-tiered preference of correctness and exclusivity as in \cite{ST05},
where in particular the correctness dimension is Boolean, i.e.~either ``correctly computed" or not.
As a result, these works fall into the category of 
``implementing cryptographic protocols with rational agents"
rather than the more game-theoretic topic ``informational mechanism design" which we address in this paper.
%``designing mechanisms for a cooperative game where the utility functions are in terms of information".

%Others
There is another line of work \cite{MNT09,NOS12} on mechanism design with privacy-aware agents who care about their privacy rather than exclusivity. 
The consideration of privacy is relevant in many applications but technically orthogonal to what we study in this work. (In our work, the privacy of the inputs is only a tool towards limiting the loss of utility due to the exclusivity preference, not a goal in itself).

%Missing
%\cnote{Add few words and comparison with ~\cite{DBLP:conf/aaai/ChenNW15} and~\cite{DBLP:conf/innovations/AzarGP16}}

The recent works of \cite{DBLP:conf/aaai/ChenNW15} and \cite{DBLP:conf/innovations/AzarGP16}
investigate non-monetary mechanisms for cooperation among competing agents.
However, an essential difference is that they consider a sequential delivery of outputs to different agents, such that the utility function is not merely in terms of information but also depends on the time or order when the output is delivered.
For example, the ``treasure hunting problem" in \cite{DBLP:conf/aaai/ChenNW15} is in particular very similar to the multiparty set intersection problem, except that in treasure-hunting only the first agent finding the common element gets positive utility while all others get zero.

\section{Multiparty Set Union}

Let $N = \{1, \ldots, n\}$ be a set of players. There is a universe $\mathcal{U}= \{u_1, \ldots, u_m\}$ of possible numbers, from which each player $i$ owns a subset $S_i \subseteq \mathcal{U}$ that is private to the player. The goal of the players is to obtain more elements of the universe from other players by sending elements from their own set in exchange.

\medskip

We study the problem of designing mechanisms that incentivize the participants to share their information with each other.
A mechanism $\mathcal{M}$ will take as input from each player $i$ a set $x_i \subseteq S_i$ and output a vector $\vec{y} = \mathcal{M}(\vec{x})$, so that the $i$-th entry of this vector contains the set received by player $i$ after the exchange. 
\medskip

\emph{Strategies}. The strategy of a player is the set it submits to the mechanism.  Players can hide elements (i.e. submit a strict subset of their true set), but not submit elements they don't actually have. A special case we will study in more depth is when the strategies of the players are ``all-or-nothing'', i.e. $x_i \in \{\emptyset, S_i\}$. The input of each player to the mechanism is sent through a private authenticated channel to the center.

\medskip

\emph{Utility}. We say the ``information benefit'' that player $i$ receives from sending his set $S_i$ to the mechanism is the number of new elements that $i$ obtains from the exchange: $v_i(\vec{x}) = |\mathcal{M}_i(\vec{x}) \setminus x_i|$. The \emph{utility} of the player is then defined as the minimum difference between his own information benefit and that of any other player, formally given by $
u_i(\vec{x}) = v_i(\vec{x}) - \max_{j \in N \setminus \{ i\}} v_j(\vec{x}).
$

The intuition is that each player wishes to learn as much as possible while maintaining exclusivity, which is captured by minimizing the amount of information obtained by the other players. This utility function is closely tied with the notion of envy as it compares the value for a player with the maximum value of any other player and the aim is to compute outcomes that are (approximately) envy-free.

\medskip

\emph{Incentive compatibility and Efficiency}. We are interested in mechanisms that incentivize players to share their information and will say that a mechanism is \emph{truthful} if truth telling is a dominant strategy for each player regardless of the strategies of the other players. 
%A mechanism will be said to be \emph{individually rational} if the utility of a player is (weakly) higher by participating than by not participating. 
An allocation (outcome) is \emph{Pareto efficient} (or \emph{Pareto optimal}) if there is no other outcome where at least one agent is strictly better off and nobody is worse off.

\medskip

\emph{Fairness}. Some of our mechanisms also satisfy fairness and the fairness notions we consider are \emph{symmetry} and \emph{strong dominance}. Symmetry requires that if multiple agents report inputs of equivalent quality, then they get the same amount of information benefit (and so the same utility). 
 Strong dominance stipulates that if the information reported by an agent is inferior to the information reported by another agent under some partial order, then the result sent to the first agent is also (weakly) inferior to the result sent to the second agent under that order.

%The utility of a player $i$ from a pairwise exchange with another player $j$ is the number of elements obtained after 

\begin{comment}
In this section, we investigate the Multiparty Set Union (MSU) problem. 
Note that we study this problem in a setting where every agent gets a customized output,
and the utility function incorporates exclusivity in the single-envy manner.

\subsection{The formal problem setting of MSU} %Multiparty Set Union}
\label{subsec:3psu}
There is a universe $U=[m]=\set{1,2,\dots,m}$ of $m$ elements and the type set $B=\set{S\st S\subseteq U}$.
Every agent $i$ has a private type $a_i\in \type$ and an action set  $\strategy_i=\set{a_i,\emptyset}$, where $\emptyset$ denotes non-participation and $\strategy=B\union\set{\emptyset}=B$.
The value of a set $S$ is naturally defined as its size, i.e.~$v(S)=\abs{S}$.
And the single-envy function $\envy(v_{-i})=\max_{j\ne i}v_j$ with $\lambda=1$.

The mechanism computes $y=(y_1,\dots, y_n)=f(x_1,\dots,x_n)$ and sends $y_i$ to $A_i$ respectively.
Thus, the information benefit of $A_i$ is $v_i= v(y_i\union a_i)-v(a_i)=\abs{y_i\backslash a_i}$, 
and the utility of $A_i$ is $u_i=v_i-\max_{j\ne i}v_j$.
Note that $v_i$ is exactly the number of new elements that $A_i$ learns from the  result $y_i$ customized by the principal.
\end{comment}

\subsection{Two Players} \label{sec:union2}

As a warm-up, we provide a simple solution to the exchange problem for $n=2$ players.

\begin{proposition}
There is a truthful polynomial time mechanism for the set union problem between two players.
\end{proposition}
\begin{proof}
Without loss of generality, we can assume the set owned by the second player is larger: $|x_1|\le |x_2|$. Let $y_2=x_1 \cup x_2$ and $y_1= x_1 \cup y'_1$, where $y_1'$ is a set chosen so that $y'_1\subseteq x_2 \setminus x_1$ and $|y'_1|=|x_1 \setminus x_2|$.
%
%let $f(x_1,x_2)=(y_1,y_2)$ such that $x_1\subseteq y_1\subseteq (x_1\union x_2)$, $x_2\subseteq y_2\subseteq (x_1\union x_2)$ and 
%$|y_1\bs x_1|=|y_2\bs x_2|=\min\set{|x_1\bs x_2|, |x_2\bs x_1|}$. 
Then players $1$ and $2$ can fairly exchange their exclusive elements until one of them has used up his exclusive elements. Note this type of exchange performed over multiple rounds can in fact be done in an atomic way by the principal.

%For the sake of exposition here and below we will describe the exchange as if it was happening in different rounds. Of course this is only an intuitive explanation, and the exchange is performed in an atomic way by the principal $T$.)
%\cnote{When preseting the paper, many were confused by this idea of "time", therefore I added a remark here. Feel free to explain differently.}
It is immediate that this mechanism ensures both agents get the same information benefit: $v_1=v_2=| x_1 \setminus x_2| \ge 0$ and it is weakly dominant for them to report their true information.
\end{proof}

\subsection{Three Players} \label{sec:union3}

For three players the problem becomes more subtle, as the mechanism must specify the order of pairwise exchanging, the number of exchanged elements, and, more importantly,
which elements are exchanged. We have the following theorem.

\begin{theorem} \label{thm:3set}
There is a truthful polynomial time mechanism for set union among $n=3$ players.
\end{theorem}
\begin{proof}%[Proof of Theorem \ref{thm:3set}]
The theorem will follow from the construction in Mechanism \ref{mech: 3psu}. See Figure 1 for a high level intuition of how the mechanism works.

Mechanism~\ref{mech: 3psu} starts by removing the common elements among all three parties, since these elements will not affect the exchange; these elements are denoted by the set $z_0$.
Then we consider the three pairwise intersections, from which the players can exchange a number of elements bounded by the smallest intersection i.e.~$s=\min\set{|x_1\intersect x_2|,|x_2\intersect x_3|,|x_3\intersect x_1|}$. Note that at the end of these exchanges at least one of these three intersections will be ``used up''. Therefore we assume w.l.o.g.~that after this step $x_2\intersect x_3=\emptyset$ and $\abs{x_2}\ge \abs{x_3}$.
Now we have reduced the original problem to a setting where there is no common intersection and only two pairwise intersections are non-empty, namely $x_1\intersect x_2$ and $x_1\intersect x_3$. 

Let $x_2, x_3$ be partitioned into $x_2=x'_2\union x''_2, x_3=x'_3\union x''_3$ where $x'_2=x_2\intersect x_1, x''_2=x_2\bs x_1$, and $x'_3=x_3\intersect x_1$, $x''_3=x_3\bs x_1$. 
The intuition will be that elements in $x'_2$ should be used to exchange elements in $x''_3=x_3\bs(x_1\union x_2)=x_3\intersect\overline{x_1}\intersect \overline{x_2}$, and similarly $x''_2$ for $x'_3$. 

Next we discuss how the exchanging looks like in different situations; again see Figure~\ref{fig: 3p} for a visual depiction.

\begin{itemize}
\item Case 1: $|x'_2| \ge |x''_3|$ and $|x''_2| \ge |x'_3|$. This is the simplest case, where we can simply make player $3$ exchange all elements in $x_3=x''_3\union x'_3$ with both players $1$ and $2$ for an equal amount of elements in $z\subseteq x'_2$ and $w\subseteq x''_2$ respectively.
Then, player $3$ used up all its elements and the problem reduces to the two-party case between
players $1$ and $2$ with remaining elements in $\left( x_1\bs x'_3, x_2\bs w \right)$.
\item Case 2: $|x''_3|>|x'_2|$ and $|x''_2|> |x'_3|$. Then player $2$ uses $\abs{x'_3}$ many elements in $x''_2$, denoted by $w$, to exchange all elements in $x'_3$ with players $1$ and $3$, and by symmetry player $3$ uses  $\abs{x'_2}$ many elements in $z$ to exchange $x'_2$ with
players $1$ and $2$.
After this exchange all the three players may have some elements left, but these are all exclusive elements, so the problem reduces to the easy case 
% after the exchange of  $x'_2$ for  $z\subseteq x''_3$, and $x'_3$ for $w\subseteq x''_2$.
of three party with disjoint elements $\left( x_1\bs (x'_2\union x'_3), x''_2 \bs w, x''_3 \bs z \right)$. Then the mechanism exchanges a number of elements equal to
$\min\set{\abs{ x_1\bs (x'_2\union x'_3)}, \abs{x''_2 \bs w}, \abs{x''_3 \bs z}}$,
further reducing the problem to the two-party case.
\item Case 3: $|x''_3|<|x'_2|$ and $|x''_2|< |x'_3|$.
In this case player $2$ uses $x''_2$ in exchange for $\abs{x''_2}$ many elements in $z\subseteq x'_3$, and, by symmetry, player $3$ uses $x''_3$ to exchange $\abs{x''_3}$ many elements in $w\subseteq x'_2$.
%After exchanging $x''_2$ for $|x''_2|$ many elements in $z\subseteq x'_3$ and symmetrically $x''_3$ for $|x''_3|$ many elements in $w\subseteq x'_2$,
After such an exchange the problem reduces to three parties with  $\left( x_1\bs (w \union z), x'_2 \bs w, x'_3 \bs z \right)$.
% Notice that now player $1$ has full information (since $x_1\supseteq x'_2$ and $x_1\supseteq x'_3$) and wants to terminate the exchange immediately,
% whereas players $2$ and $3$ may still want to exchange their remaining elements in $x'_2 \bs w$ and $x'_3 \bs z$ (the remaining exclusive elements of $1$ in $x_1\bs(x_2\union x_3)$, in case there are any, will never be used anyway). 
%
Finally, for a better social wellfare, we let player $2$ and player $3$ run a na\"{i}ve two-player exchange protocol with their remaining elements in $x'_2 \bs w$ and $x'_3 \bs z$.
This is not optimal for player $1$, who has already collect full information and wants to terminate the exchange immediately. 
However, player $1$ cannot prevent such exchange between players $2$ and $3$ anyhow.
% \textbf{Then, we have to decide whether the mechanism should let players $2$ and $3$ exchange their remaining elements,
% given that player $1$ prefers to stop. It turns out that this is not a critical problem since it does not affect the individual rationality of the mechanism.}
% \textcolor{red}{In the bold sentence it's not clear if at the end we exchange or not. The end of this case should be sharper so that it's clear what we do.}
\end{itemize}
%even if in the \emph{flexible model} where we allow partial participation, i.e.~$\strategy_i=\set{x\st x\subseteq x_i}$.
% \footnote{For social welfare maximization, it is reasonable to let players $2$ and $3$ continue exchanging regardless of player $1$'s opinion.}
%\cnote{...to here I think it's unclear. Also, I think it would be better to remove the mention to the flexible model alltgoether?}

Mechanism~\ref{mech: 3psu} guarantees individual rationality because every round of exchange in its process is ``fair'' and ``necessary''.
Every round is fair in the sense that all participants of that round get equal benefits --- each of them gives out some elements in exchange for more new elements. 
Every round of such exchange is necessary because each element appears in at most one round, i.e.~the mechanism does not reuse previously exchanged elements.
Therefore, a player that hides elements would suffer a loss lower bounded by the number of private elements that could have been traded, which is indeed a natural upper bound for the loss of others.
%Therefore, if an agent withdraws participation (which is the only possible deviation in the \emph{all-or-nothing} model), they would suffer a loss lower bounded by the number of private elements that could have been traded, which turns out to be a natural upper bound for the loss of others. (Note that this argument would also work to address the more \emph{flexible} case of agents who are allowed to report any subset of their input -- instead of having only the options of reporting their set truthfully or not participating at all. This deviation strategy would be possible in e.g., a model where a certification authority certifies each element in the agents set individually, so that it is impossible for an agent to add extra elements to their sets, but it is possible to remove elements in an undetectable way).
\end{proof}

\bigskip

\begin{algorithm}[H]\label{mech: 3psu}
	\caption{Three Party Set Union}
	
	\SetAlgoLined
	\KwIn{  Set $x_i\subseteq \mathcal{U}$ for each player $i$ }
	\KwOut{ Set $y_i \subseteq \mathcal{U}$ for each player $i$ }
	
	\medskip
	
	%	\ForEach{$A_i$} {$y_i\gets x_i$;}
	
	$z_0 = x_1\intersect x_2 \intersect x_3$
	
	\ForEach{player $i$} {$y_i = x_i$
		
		$x_i = x_i\bs z_0$} 
	
	\smallskip
	
	$s = \min\left\{\abs{x_1\intersect x_2}, \abs{x_2\intersect x_3}, \abs{x_3\intersect x_1}\right\}$  \tcc*[h]{\emph{W.l.o.g., $\abs{x_2} \ge \abs{x_3}$ and $s=\abs{x_2\intersect x_3}$}}
	
	$z_1 = x_2\intersect x_3$
	
	Select arbitrary sets $z_2\subseteq x_3\intersect x_1$ and $z_3\subseteq x_1\intersect x_2$ of sizes $\abs{z_2}=\abs{z_3}=s=\abs{z_1}$
	
	\smallskip
	
	\ForEach{player $i$} {$y_i = y_i\union z_i$
		
		$x_i = x_i\bs \left(z_1\union z_2 \union z_3\right)$}
	
	$x'_2 = x_2\intersect x_1$; $x''_2 = x_2\bs x_1$
	
	$x'_3 = x_3\intersect x_1$; $x''_3 = x_3\bs x_1$
	
	\smallskip
	
	$(y_1', y_2', y_3') = (\emptyset, \emptyset, \emptyset)$ \tcc*[h]{\emph{Sets that will store elements from recursive calls, if any.}}
	
	\medskip

	\uIf{$|x'_2| \ge |x''_3|$ and $|x''_2| \ge |x'_3|$}{\tcc*[h]{\emph{Case 1}}
		\smallskip
		
		Select arbitrary sets $z \subseteq x'_2$ and $w \subseteq x''_2$ of sizes $\abs{z}=\abs{x''_3}$ and $\abs{w}=\abs{x'_3}$
		
		$y_2 = y_2 \union x_3$
		
		$y_3 = y_3 \union z \union w$
		
		$y_1 = y_1 \union w \union x''_3$
		
		$(y_1', y_2') = \mbox{\textsc{TwoPartySetUnion}}\left( x_1\bs x'_3, x_2\bs w \right)$
			
}
\smallskip
	\uElseIf{$|x''_3|>|x'_2|$ and $|x''_2|> |x'_3|$}{\tcc*[h]{\emph{Case 2}}
		\smallskip

		Select arbitrary sets $w \subseteq x''_2$ and $z\subseteq x''_3$ of sizes $\abs{w}=\abs{x'_3}$ and $\abs{z}=\abs{x'_2}$
		
		$y_2 = y_2 \union x'_3 \union z$
		
		$y_3 = y_3 \union x'_2 \union w$
		
		$y_1 = y_1 \union z \union w$
		
		$(y_1', y_2', y_3') = \mbox{\textsc{ThreePartyDisjointSetUnion}}\left( x_1\bs (x'_2\union x'_3), x''_2 \bs w, x''_3 \bs z \right)$
	}
\smallskip
	\uElse{\tcc*[h]{\emph{Case 3: $|x''_3|<|x'_2|$ and $|x''_2|< |x'_3|$}}\\
		\smallskip
		Select arbitrary sets $w \subseteq x'_2$ and $z\subseteq x'_3$ of sizes $\abs{w}=\abs{x''_3}$ and $\abs{z}=\abs{x''_2}$
		
		$y_2 = y_2 \union x''_3 \union z$
		
		$y_3 = y_3 \union x''_2 \union w$
		
		$y_1 = y_1\union x''_2 \union x''_3$
		
		$(y_2', y_3') = \mbox{\textsc{TwoPartySetUnion}}\left(  x'_2 \bs w, x'_3 \bs z \right)$	
	}
\smallskip
	\ForEach{player $i$}{$y_i = y_i \cup y_i'$ \tcc*[h]{\emph{Add elements obtained from recursive calls, if any, to the final set for each player.}}}
			
	return $\left(y_1,y_2,y_3 \right)$
	
\end{algorithm}
\bigskip
\begin{figure}[htb]
	\caption{\large{\textcolor{black}{\textbf{\emph{High level idea for Mechanism 1.}}}}}
\bigskip
	\centering
	\includegraphics[scale=0.38]{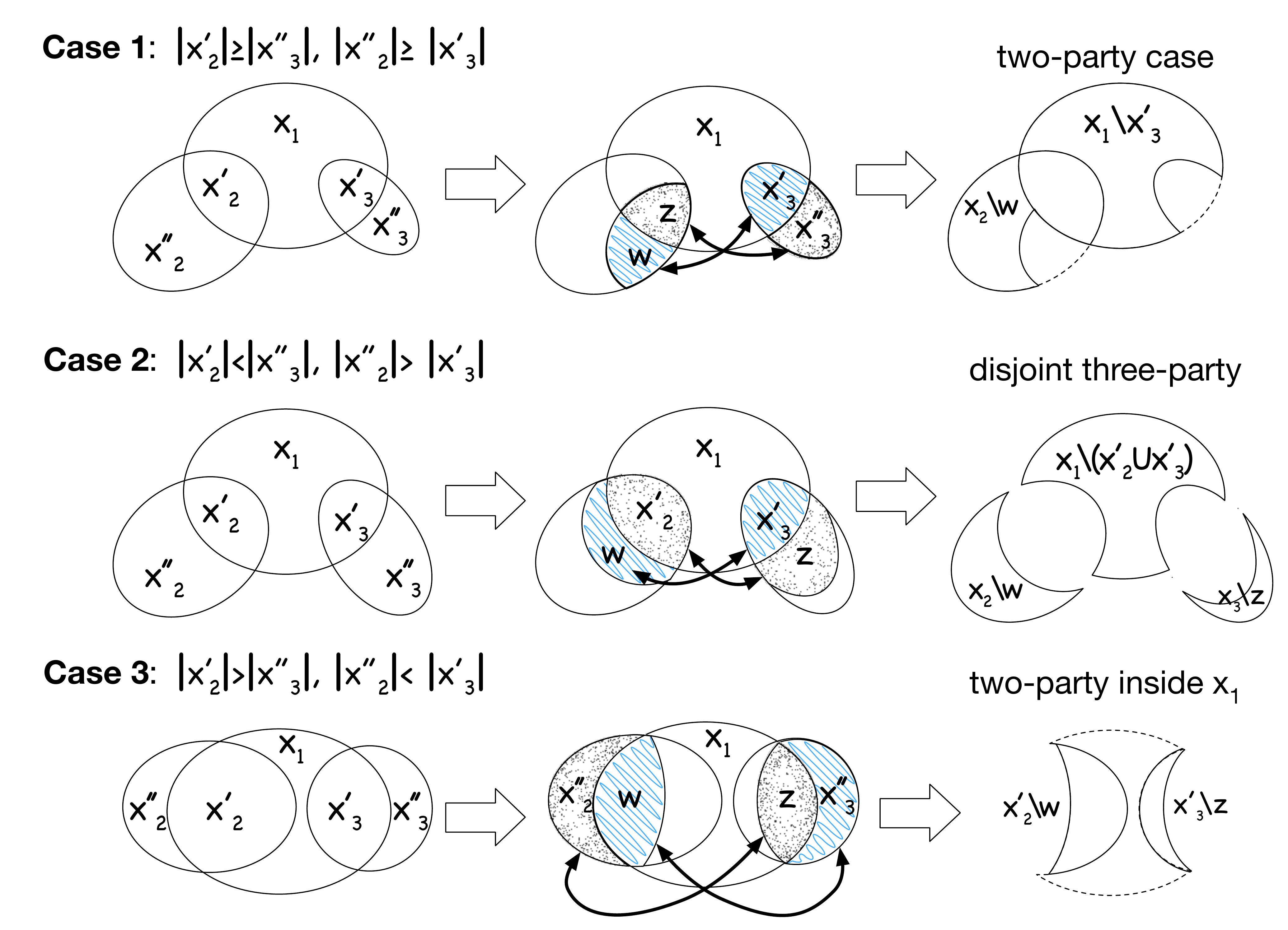}
	\label{fig: 3p}
\end{figure}

We note that Mechanism~\ref{mech: 3psu} is not Pareto efficient since the reduced problem (that is solved in the recursive call) is dealt with in a na\"{i}ve way.
For example, consider the last step in Case 1 in Figure~\ref{fig: 3p} i.e., after the problem has already been reduced to the two-party case. Here we could let players $1$ and $2$ exchange their remaining elements without player $3$.
This could be seen as fair, since player $3$ does not contribute new elements in those rounds.
However, this procedure does not achieve Pareto efficiency, for that we can improve the social welfare by giving player~$3$ some extra elements and, for sufficiently small number of elements,
the utilities of players $1$ and $2$ would not change and the solution would ensure that players $1$ and $2$ remain truthful as before.

\medskip

However, Mechanism~\ref{mech: 3psu} can be turned into a Pareto efficient mechanism as follows.

\begin{theorem}
	There is a truthful and Pareto efficient mechanism for set union among $n=3$ players.
\end{theorem}
\begin{proof}
For case $1$, consider the last step (of this case) in the execution of the mechanism (see Figure~\ref{fig: 3p}), after the problem has already been reduced to the two party case. 
Mechanism~\ref{mech: 3psu} can be modified to 
assign randomly selected extra elements to player $3$ so that $|y_3\bs x_3|=v_1 =v_2$ (recall that Mechanism~\ref{mech: 3psu} ensures $v_1=v_2$).
	This modification achieves Pareto efficiency since any further improvement on social wellfare will decrease the utility of player $1$ or player $2$, who already get all elements and cannot get more information benefit. 
	Now we prove that the above modification also preserves truthfulness.
	This is immediate for players $1$ and $2$ but requires the following observations to see that it continues to hold from the point of view of player $3$:
	\begin{itemize}
\item each of player $3$'s exclusive elements in $x_3''$ leads to the same amount of marginal benefit to player $3$ as to players $1$ and $2$, i.e. it is used to exchange for either one element in $x_2'$, which will not be exchanged between players $1$ and $2$, or two elements when
players $1$ and $2$ exchange elements in $x_2\bs w$ and $x_1\bs(x_2'\union x_3')$, respectively.
\item all of player $3$'s elements in $x_3'$ do not affect others' information benefits; however, such elements can help player $3$ since they might prevent the player from receiving some previously known element as the extra benefit. 
\end{itemize}

We note that the same modification also works to ensure Pareto efficiency in case $2$, while case $3$ already ensures a Pareto efficient exchange. Thus there exists a truthful mechanism that is Pareto efficient.
\end{proof}

\subsection{Any Number of Players}

We observe that the three player mechanism above relies on a complex analysis that depends on the different intersection sets. The number of intersections increases exponentially as $n$ grows and we leave open the question of whether it is possible to achieve an analogue of Mechanism~\ref{mech: 3psu} for more than three players.

%[\textcolor{red}{Simina's note: I'm confused actually about the difference between individual rationality (or stable cooperation) and just submitting the empty set. How is submitting the empty set different from not participating?}]

\bigskip

\noindent \textbf{Open Problem 1.}
\emph{Is there a truthful polynomial time mechanism for set union for any number of players? Are there randomized such mechanisms?}

\medskip

In the following, we design a truthful mechanism for set union for any number of players in the special case where each player can either submit its whole set or the empty set, i.e. cooperate or not. We call this the ``all-or-nothing'' model and our main result in this section is the following theorem.

\begin{theorem} \label{thm:multiparty}
	There is a truthful, Pareto efficient, and welfare maximizing mechanism
	% \footnote{Welfare maximization is achieved by maximizing over all the Pareto efficient outcomes.}
	 for set union among any number $n$ of all-or-nothing players. The mechanism runs in polynomial time for any fixed $n$. 
%	Mechanism~\ref{mech: npsu} guarantees stable cooperation and its output is the Pareto optimal solution with maximal social welfare for Multiparty Set Union problem in the all-or-nothing model.
%This mechanism is also symmetric and strongly dominant.
\end{theorem}

\begin{comment}
In this section we provide a mechanism for set union with $n>3$ agents. We start by noting that the approach used in the three-party mechanism discussed in the previous section (i.e., ``exchanging'' elements from mutually exclusive subsets) does not seem to scale for multiple parties. First, the number of intersections grows exponentially with the number of agents, thus making it difficult to specify and analyse which exchanges are to be performed. Second, as described in Remark~\ref{remark1}, the approach used in the previous section allows for ``free riding'' (e.g., participating in an exchange round without contributing any new information to the other participants), and therefore it is more difficult to guarantee participation when there is a chance to get information benefits without contributing.%\cnote{I am not completely sure what you mean with free auditing}

% it seems plausible to consider an ``exchange-based" mechanism.
%However, since our modification in 
%Also, it becomes more difficult to specify and analyze the exchange details for exponentially many possible categories of %elements.

\end{comment}

\medskip

To prove the theorem we develop several lemmas.

%Before describing our multiparty mechanism, let us first discuss the properties and necessary conditions of Pareto optimal solutions.

\begin{lemma}\label{observation1}
Let $M$ be any mechanism for set union for $n$ all-or-nothing players. If the mechanism ensures that in every execution the information benefit of each player is the same as that of any other player, then the mechanism is truthful and Pareto efficient.
%	Assigning a unified information gain $v_i=V$ to every agent $A_i$ is always a stable Pareto optimal solution in the all-or-nothing model.
\end{lemma}
\begin{proof}
	Consider the outcome of any execution of $M$ and suppose there is a value $V$ so that the information benefit of each player $i$ is $v_i = V$.
	An assignment with this property must satisfy truthfulness on this instance for all-or-nothing players since every player $i$ gets utility $u_i=v_i-\max_{j\ne i}v_j=0$ when participating and non-positive utility when not participating. 	
	
	%By Lemma~\ref{lemma:PO necessary condition} such a solution is always Pareto optimal.
	By definition, $u_i=v_i-\max_{j \in N \setminus \{i\}} v_j$ for every player $i$. Then we have
	$\sum_{i \in N} u_i =\sum_{i \in N} v_i -\sum_{i \in N} \max_{j \in N \setminus \{ i\}} v_j\le 0$
	It follows that for any player $i$ with $u_i>0$ there must be another player $j$ with $u_j<0$, so no Pareto improvement is possible for this solution.
\end{proof}

In particular, for $V=0$, assigning $y_i=x_i$ to  every agent $A_i$ is always a stable and Pareto optimal solution, since there is no difference between participating or not.
However, this trivial solution is obviously the worst in social welfare.

To achieve maximum social welfare among all truthful and Pareto optimal solutions, we notice that by increasing the unified gain $V$ the social welfare grows respectively while preserving truthfulness and Pareto optimality.
However, $V$ cannot be arbitrarily large since there is a systematical upper bound for the possible gain of agents with large sets, i.e. $v_i =|y_i\bs a_i|\le \big| \left(\union_i x_i\right)\bs a_i \big|$ since $y_i\subseteq \union_i x_i$.

\medskip
The full characterization of Pareto optimal solutions for the Multiparty Set Union is summarized next.

\begin{lemma}\label{lemma:PO condition}
	Let $M$ be a mechanism for set union among $n$ players. Suppose that on some execution, the information benefit of each player $i$ is $v_i$.
	
	Let $i=\argmax_{j=1}^n v_j$ and $V=\max_{j \neq i} v_j$.
	Then the allocation is Pareto optimal if and only if for every $j\ne i$, player $j$ gets an information benefit of
	$v_j=\min\set{V, \big|\bigcup_k x_k \bs a_j\big|}.$
	That is, every player $j$, except player $i$ who gets the maximum, gets $V$ many new elements (unless there are fewer than $V$ that can be assigned to that agent).
\end{lemma}
\begin{proof}
	Suppose that on the given instance there is a player $j$ that gets a benefit bounded by $v_j<V$ and $v_j< \big|\union_k x_k \bs a_j\big|$.
	Then the mechanism can be modified on this instance to increase $v_j$ by one to achieve higher social welfare. 
	This is a Pareto improvement since no other players care about it: player $i$ only cares about the one who already receive $V$, while all the other players envy player $i$'s information gain.
	
	On the other hand, every such allocation achieves Pareto optimality since every agent's information gain is either impossible to improve or envied by another agent and hence cannot be improved without decreasing the utility of any other agent.
\end{proof}

In the rest of this section, we focus on determining the unified value $V$ that maximizes social welfare while preserving truthfulness and Pareto efficiency.

\bigskip

\begin{algorithm}[H]\label{mech: npsu}
	\caption{Multiparty Set Union}
	
	\SetAlgoLined
	\SetKwProg{Fn}{Function}{}{}
	\SetKwFunction{FRecurs}{ComputeV}%

	\KwIn{ $(x_1,x_2,\dots,x_n)$, where each set $x_i\subseteq \mathcal{U}$ is the input from player $i$.  }
	\KwOut{ $(y_1,y_2,\dots,y_n)$, where each set $y_i$ is sent to player $i$. }
	
	\medskip
	
	Fix an ordering $\pi$ of all elements in $\mathcal{U}$  \tcc*[h]{\emph{$\pi$ will be used to specify the exchanged elements}}
	%	 \tcp*{$\pi$ is used to specify the exchanged elements and guarantee strong dominance property}

	%	Sort $y$ and w.l.o.g. assume that $|y_1|\ge |y_2|\ge \dots \ge |y_n|$;
	% \uIf{$n \le 1$}{$V = 0$}
	% % 	\ElseIf{$n=2$}{
	% % 	$V \gets \min\set{|y_1\bs y_2|, |y_2\bs y_1| }$;
	
	% % %	$z \gets $
	% % 	}
	% \Else{ 
	% 	\ForEach{player $i\in[n]$} {

	% 		$z_{-i} \gets \union_{j\ne i} x_j$;
			
	% 		%			$z_{-(i,j)}\gets \union_{k\ne i,j} x_k$
			
	% 		compute	the unified value $V_{-i}$ for $A_{-i}$;
	% 	}
		
	% 	%		fix an arbitrary $j\in [n]$;
	% 	% $j\gets \argmax_j |x_j  \bs z_{-j}|$;

	% 	% \If{$|x_j| > |z_{-j}|$}{
	% 	% 	let $s_j\subset x_j\bs z_{-j}$ consist of the last $|s_j|=|x_j| - |z_{-j}|$ elements in $x_j\bs z_{-j}$ according to $\pi$;

	% 	% 	$x_j\gets x_j\bs s_j$;  \tcp{$|x_j| \le |z_{-j}|$ after reducing $x_j$}
	% 	% } 

	% 	%		$k\gets \argmax_{k\ne j} |x_k|$;
		
	% 	%		compute $V_{-j}$ for $A_{-j}$;
		
	% 	%		\ForEach{$k\in [n] $}{
	% 	%		compute $v_{k,-j}$ for the benefit of $A_k$ in the cooperation among $A_{-j}$;
		
	% 	%		compute	$V_{-k}$ for $A_{-k}$;
	% 	% and $V_{-(j,k)}$ for $A_{-(j,k)}$;

	% 	%		$\Delta'_{j,-k}\gets V_{-k}-V_{-(j,k)}$;
	% 	%		}
		
	% 	%		$\Delta_j\gets \min_{k}\set{|z_{-k}\bs x_k| + V_{-k}-V_{-j}}$;
	% 	%, |z_{-j} \bs x_j|}$;
	% 	%		$\Delta_j\gets \min\set{\min_{k\ne j}\set{|z_{-k}\bs x_k|-v_{k,-j}+ \Delta'_{j,-k}}, |z_{-j} \bs x_j|}$;
		
	% 	%		$V\gets V_{-j}+\Delta_j$;
		
	% 	$V \gets \min\set{\min_{k\in[n]}\set{|z_{-k}\bs x_k| + V_{-k}},\max_{k\in[n]}|\union_{i\in [n]}x_i\bs x_k|}$; 
	% 	%\cnote{what is $k$ here?}

	% }	
	\smallskip
	$u = \bigcup_{i=1}^n x_i$

	$V =$ \FRecurs{$x_1,\dots,x_n$} \tcc*[h]{\emph{the function \FRecurs is defined below}}
\smallskip

	\ForEach{player $i\in[n]$} {
		$v_i = \max\set{V, \left| u \bs x_i\right|}$
		
		Let $r_i$ be the set of first $v_i$ elements in $z_{-i}\bs x_i$ according to $\pi$
		
		%	let $r_i\subseteq \left(z_{-i}\bs x_i\right)$ such that $|r_i|=v_i$;	
		
		$y_i =  x_i\union r_i$
	}

	\KwRet{$\left(y_1,y_2,\dots, y_n \right)$.}

	\bigskip

	\Fn{\FRecurs{$x_1,\dots,x_n$}}{
		\uIf{$n \le 1$}{\KwRet{$0$}}

		\ForEach{player $i\in[n]$} {

			$z_{-i} = \bigcup_{j\ne i} x_j$
			
			%			$z_{-(i,j)}\gets \union_{k\ne i,j} x_k$
			
			$V_{-i} =$ \FRecurs {$x_{-i}$}
		}
		$V = \min\set{\min_{k\in[n]}\set{|z_{-k}\bs x_k| + V_{-k}},\max_{k\in[n]}|z_{-k}\bs x_k|}$
		%\cnote{what is $k$ here?}

	%	}

	}{\KwRet{$V$}.}

\end{algorithm}
\bigskip

We can now prove Theorem \ref{thm:multiparty}, showing that Mechanism~\ref{mech: npsu} satisfies the desired properties.

\begin{proof}[Proof of Theorem \ref{thm:multiparty}]
	%Suppose that $V=|z_{-j}\bs x_j|+V_{-j}$ for the specific $j$.
	% The intuition of Mechanism~\ref{mech: npsu} is that we let the marginal contribution of $A_j$ to $V$ be bounded by both the maximal information gain that $A_j$ can get from a stable cooperation and the new elements brought by $A_j$'s participation -- 
	% either from $A_j$'s exclusive elements $x_j\bs z_{-j}$ or his influence on other agents whose elements are not fully used (e.g.~those in $x_k\bs z_{-(j,k)}$ when $|x_k|>|z_{-(j,k)} |$).
	%Because $A_j$ must get at least $\Delta_j$ many new elements in any stable cooperation,
	%and the $\ell_j$ new elements brought by $A_j$'s participation is also a natural upper bound for his marginal contribution to $v$.
	
	First we prove that Mechanism~\ref{mech: npsu} is truthful. For this, note that each player $k$ prefers to submit its set rather than the empty set since $v_{k}=\min\set{V, |z_{-k}\bs x_k|}\ge V-V_{-k}$, where $|z_{-k}\bs x_k|\ge V-V_{-k}$ follows the definition of $V$.

	Next we prove that Mechanism~\ref{mech: npsu} also achieves the maximal social welfare and Pareto optimality simultaneously.
	By Lemma~\ref{lemma:PO condition} the output of Mechanism~\ref{mech: npsu} is Pareto optimal for every fixed $V$.
	However, not all possible Pareto optimal solutions are covered by solutions in that form, i.e. it remains possible that exactly one agent gets more than $V$.
	In what follows we first eliminate the possibility that one player gets more than $V$, then we prove the $V$ specified in our mechanism is already optimal.
	
	Consider a Pareto optimal and stable solution (i.e. where the players cooperate by submitting their entire set) where each player $i$ gets benefit $v_i$, such that $v_1=V>V'$ and $v_i=\min\set{V',|z_{-i}\bs x_i|}$ for $i\ge 2$. 
	Let us change it into another solution where $v_i=\min\set{V,|z_{-i}\bs x_i|}$ for all agents.
	Obviously this new solution has (weakly) better social welfare and it remains incentive compatible for every player except player $1$, since their utilities either increase or remain as before. 
	Then we consider player $1$, which is the only worse off player under such a change.
	Indeed, player $1$ will still submit its whole set because the player cannot increase his utility by withdrawing his participation, when player $1$ suffers a loss of $V$ while the information benefit of every other player decreases by at most $V$ since $v_i\le V$.

	Finally, we argue that the value $V$ specified in Mechanism~\ref{mech: npsu} cannot be effectively increased while preserving truthfulness.
	%, when $V< \max_j|\union_i x_i\bs x_j|$. 
	Suppose that $V< \max_k|\union_i x_i\bs x_k|$ and the value of $V$ is increased to $V'>V$, i.e. this increment is meaningful for at least one player.
	Then by definition of $V$, there exists $k\in [n]$ such that
	$V'>V=|z_{-k}\bs x_k| + V_{-k}$, %\ge v_k + V_{-k}$,
	and hence player $k$ prefers not to participate since $v_k\le |z_{-k}\bs x_k| <V'-V_{-k}$.
	Thus the value $V$ in Mechanism~\ref{mech: npsu} is truthful and the mechanism outputs a Pareto efficient solution with optimal social welfare as required.

To argue the runtime, note that the number of recursive calls only depends on the number of players $n$. The set operations inside each call run in polynomial time since they are simple set differences and unions.
\end{proof}

We note that in fact the mechanism guarantees that identical players are treated identically (symmetry)
and that players that contribute more also receive more from the mechanism (strong dominance).

\begin{corollary}
	There is a truthful, Pareto efficient, welfare maximizing, symmetric, and strongly dominant mechanism for set union among any number $n$ of all-or-nothing players.
	%	Mechanism~\ref{mech: npsu} guarantees stable cooperation and its output is the Pareto optimal solution with maximal social welfare for Multiparty Set Union problem in the all-or-nothing model.
\end{corollary}
\begin{proof}
	These properties are met by Mechanism~\ref{mech: npsu}.
	
	The symmetry property is trivial in this mechanism, since the information benefit of a player $i$ is 
	$$v_i=\max\set{V, |z_{-i}\bs x_i|}=\max\set{V, |\union_{j\in [n]} x_j| - |x_i|}$$ 
	which only depends on the size of its input i.e.~$|x_i|$.

	The strong dominance property is guaranteed since the reward $r_i$ to $A_i$ is exactly the first $v_i$ elements in $\union_{j\in [n]} x_j\bs x_i$ (according to the fixed ordering $\pi$), so that there must be $y_k\subseteq y_i$ for every player $k$ with $x_k\subseteq x_i$.
\end{proof}

\section{Beyond Union: Intersection and Average}

Moving beyond the multiparty set union problem, we suggest two other set problems where the agents own data points and wish to share them.
\bigskip

\noindent \textbf{One Dimensional Search}. The first problem is interval intersection, where each agent owns an interval in $\Re$ and the goal is to find a point in the intersection of all the sets. A high level scenario motivating this problem is that there is a group of people trying to find a gold mine situated at some location $t$, and each person $i$ has an estimate of where the gold mine is, given by a radius $d_i$, that is, the set $[t - d_i, t + d_i]$. The players would like to merge their estimates to get a better idea of where the gold mine is situated, but the challenge is that agents that have very good estimates (i.e. small sets) will not learn much from those with worse estimates (i.e. larger sets). Our main result for this problem is the following theorem, the details for which can be found in Appendix~\ref{sec:toy}.

\begin{theorem}
	There is a truthful polynomial time mechanism for interval intersection among any number $n$ of all-or-nothing players.
\end{theorem}

\bigskip

\noindent \textbf{Set Average}. The second problem is taking the average of a set that is distributed among the players.

\begin{theorem}
	There is a truthful polynomial time mechanism for the average point problem among any number $n$ of all-or-nothing players.
\end{theorem}

The mechanism is given in Appendix~\ref{sec: average point} and works by returning to each player that participated (i.e. submitted its point) the average of the set of points submitted, and nothing to the players that did not participate.

\medskip

In Appendix~\ref{app:general} we also discuss mechanisms for sharing problems where the value used to estimate the benefit of a player is more general.

\section{Discussion}

Aside from our concrete open questions, the directions of generalizing the results to richer strategy spaces, allowing randomization, and more general utility functions are also interesting. And at a higher level, the problem of understanding the interplay between having value for information and having value from exclusivity remains largely open. 

\paragraph*{Acknowledgments.} Research supported by: the Danish Independent Research Council under Grant-ID DFF-6108-00169 (FoCC); the European Union's Horizon 2020 research and innovation programme under grant agreement No 731583 (SODA).

\bibliographystyle{alpha}

\bibliography{itcsbib}

\appendix

\newpage

\section{One Dimensional Search} 
\label{sec:toy}

In this section we study a problem that is also related to sharing information captured through sets, except the goal is to find a point in the intersection of all the sets. A high level example of such a problem is the following.  

\bigskip

\noindent \textbf{Gold Mine Example}. Suppose there is a group of people and everyone is interested in finding a gold mine. The gold mine is situated in location $t$.
Everyone has some estimate of where the gold mine is $t_i$ and some uncertainty given by a radius $d_i$, i.e. each player $i$ has an interval
$[t_i - d_i, t_i + d_i]$. The players want to join their information to get a better approximation of the location of the gold mine and know that the gold mine lies in the intersection of all the estimates (sets).
However if a player $i$ knows that its radius $d_i$ is much smaller than that of another player $j$, then player $i$ knows that it won't learn much by interacting with player $j$. That is, in the worse case player $i$'s interval is contained in player $j$'s interval, so there is no information player $i$ can infer from $j$. Since player $i$ would rather not have player $j$ gain free information without receiving anything in return, the problem is to design a mechanism that incentivizes the players to learn from each other (as much as possible). 

\medskip

More formally, suppose each player $i$ has an interval in $[\alpha_i, \beta_i] \subset \Re$ and the goal is to find a point contained in all the intervals. We are promised that such a point exists.

\medskip

We solve this problem for all-or-nothing players, where a player can either cooperate by submitting its interval $[\alpha_i, \beta_i]$ or not cooperate by submitting the whole set of real numbers $\Re$. Given that the intersection point chosen is $t$, the information benefit that a player derives from learning an interval $[a, b]$ will be given by an arbitrary monotone function $v$ such that $v(a,b) = v(b - a)$ as long as $a \leq t \leq b$ and $v(a,b) = -\infty$ otherwise. Given an allocation where the benefit to each player $i$ is $v_i$, the utility of a player $i$ is $u_i = v_i - \max_{j \neq i} v_j$. 

%We first provide a mechanism for the case where the sets owned by the players are different from each other, and then generalize to the case where there %	can be identical sets.

\medskip

\begin{comment}
problem of finding a target number $t$ on the axis. 
Let $\type=\{[\alpha,\beta]\; \big| \alpha\le t \le \beta \}$ be the set of intervals containing $t$, and let $\strategy=\result=\type\union\set{\bot}$. (It is convenient to think of $\bot$ as the interval $\bot=[-\infty,+\infty]$.)
Let $v$ be any reasonable value function such that $v([\alpha,\beta])=v(\beta-\alpha)$ is a strictly monotone decreasing function of the interval size $\beta-\alpha$ as long as $t\in [\alpha,\beta]$, and $v([\alpha,\beta])=-\infty$ when $t\notin [\alpha, \beta]$.
We also assume a reasonable mechanism $f$ such that for every $A_i$ reporting $x_i=a_i=[\alpha_i,\beta_i]$, it holds that $t\in f(a_i;x_{-i})\subseteq [\alpha_i,\beta_i]$ (as long as $t\in\intersect x_{-i}$), 
%and hence always $v(f([\alpha_i,\beta_i];x_{-i}))\ge v([\alpha_i,\beta_i])$. Thus, 
which guarantees a non-negative benefit $v_i=v(f(a_i;x_{-i})) - v(a_i)\ge 0$.
At this stage we assume that the agents hold distinct inputs.
%since identical inputs should lead to identical outputs.
% all agents are sorted by their input quality, 
%i.e.~$|a_1|\le |a_2|\le \dots\le |a_n|$,
%%$|\beta_1-\alpha_1|\le |\beta_2-\alpha_2|\le \dots \le |\beta_n-\alpha_n|$, 
%and furthermore $a_i$'s are distinct.
%% (if two agents submit the same type to the center, they should always get exactly the same result in a fair mechanism). 
\end{comment}

\begin{theorem}
	There is a truthful polynomial time mechanism for interval intersection among any number $n$ of all-or-nothing players.
\end{theorem}
\begin{proof}
Consider the following mechanism. 

\medskip

\begin{algorithm}[H]\label{mech:1}
	\caption{One Dimensional Search Mechanism}
	
	\SetAlgoLined
	\KwIn{$(x_1, \ldots, x_n)$, where $x_i$ is the set submitted by each player $i$.  }
	\KwOut{$(y_1, \ldots, y_n)$, where $y_i$ is the set received by each player $i$. }
	
	\medskip
	
	\ForEach{player $i$} {$y_i = x_i$}
	
	$j = 1$; $k = 1$
	
	\ForEach{player $i$}{
		\If{$\alpha_j<\alpha_i$ \emph{\textbf{or}} $\left(\alpha_j = \alpha_i\right.$ \emph{\textbf{and}} $\left.\beta_j> \beta_i\right)$}{
			$j = i$
		}
		\If{$\beta_k > \beta_i$ \emph{\textbf{or}} $\left(\beta_k = \beta_i\right.$ \emph{\textbf{and}} $\left.\alpha_k<\alpha_i\right)$ }{
			$k = i$
		}	
	}     
	\If{$j\ne k$}{
		Select points $\beta_j',\alpha_k'$ so that $v(y_j)-v(x_j)=v(y_k)-v(x_k)>0$
		
		$y_j = [\alpha_j,\beta_j']$; $y_k = [\alpha'_k,\beta_k]$
		
	}
	
	\KwOut{ $\left(y_1,\dots, y_n \right)$}
%	Output $\left(y_1,\dots, y_n \right)$.
	
\end{algorithm}
The high level intuition is that the two players contributing to the most accurate interval are equally rewarded, i.e. these two players  exchange their information fairly (equal benefit).
Other players get no update on what they submitted.

For players that have intervals all different from each other, note that 
in Mechanism~\ref{mech:1}, the interval $[\alpha_j,\beta_k]$ tracks the intersection of all processed $x_i$'s
and eventually $[\alpha_j,\beta_k]=\Intersect_{i=1}^{n} [\alpha_i,\beta_i]$.
Then it is easy to verify that such a mechanism is truthful for the utility function we study.
To handle players that may have identical sets as input, note that Mechanism~\ref{mech:1} can be made to handle this case by assigning identical results to agents reporting identical types.
% \textcolor{red}{I removed all references to other types of utility functions. Is this enough for this proof? I would add a brief description of what the algorithm is doing intuitively.}
\end{proof}

%The formal description of the mechanism implemented by $T$ is given in Mechanism~\ref{mech:1}.
%Let $f$ be defined as follows: 

%\paragraph{The choice of envy function.}

\begin{comment}
\paragraph{Limitations of Linear Envy.} The analysis of the above mechanism reveals a possible issue with the choice of the \emph{linear-envy} function and motivates for using \emph{single-envy} instead. Recall that the linear-envy function is $ \envy(v_{-i})=\lambda \sum_{j\ne i} v_j$. Now, consider the case where there are $n'> 1/\lambda$ many agents of exactly the same type as $1$, and in a fair mechanism they get the same result $y_1$ and information benefit $v_1$ as $1$.
Now,  $u_1=v_1-\lambda \sum_{j\ne 1} v_j=(1-\lambda n')v_1-\lambda \sum_{a_j\ne a_1} v_j$ turns out to be a decreasing function in $v_1$ since $1-\lambda n'<0$,
which means that $1$ would prefer to get a worse result!
%
Such undesirable behavior happens just because the envy function is a summation, 
so that it magnifies the effect and leads to a non-constructive situation even for very small envy factor $\lambda$.
%
To avoid this we suggest to use the \emph{single-envy} function $\envy(v_{-i})=\lambda\max_{j\ne i}\set{v_j}$ where each agent $A_i$ only envies the agent among $A_{-i}$ who receives the highest payoff and ignores others, which also fits better the literal interpretation of ``envy".
When using single-envy, the mechanism is able to focus on the reported types of agents but ignores their identity,
i.e.~multiple agents of identical type are treated as a single one.
\end{comment}

\newpage

\section{Average Point Problem}
\label{sec: average point}

In this section we consider another well-known multiparty problem---the Average Point problem---where every agent is assigned a private input and they want to compute the average value of their private inputs. 
This problem is very different from set union, in the sense that the quality of private information cannot be objectively measured, since it depends on the private inputs of other agents.
As a result, we resort to a mechanism where in some sense the principal treats the agents more equally.
\medskip

Let the universe $\mathcal{U}$ be a metric space (e.g. $\mathcal{U} = \Re$ or $\mathcal{U}=\Re^2$). There are $N = \{1, \ldots, n\}$ players, so that each player $i$ has a private point $a_i \in \mathcal{U}$. 
The goal is to compute the average point $\overline{a}= \sum_{i=1}^n {a_i}/n$. 
We focus on all-or-nothing players, who either submit their point $a_i$ or nothing, the latter of which is denoted by $\bot$. 

\begin{comment}
and let the type set be $\type=\set{u \st u\in U}$.
Then, we consider $n$ agents $1,\dots,A$ and each agent $A_i$ has a private type $a_i\in \type$ which refers to a point in $\mathcal{U}$ and an action set $\strategy_i\subseteq \type$.
Every agent $i$ submits $x_i\in \strategy_i$ to the principal $T$, 
%\cnote{In the multiparty problem we don't mention the trusted party. I think it would be better if the two sections were more uniform. Also, is it called trusted party in the setting of mechanism design? Or center? Or?}
and then $T$ evaluates a function $f: U^n \to U^n$ on reported $(x_1,\dots, x_n)$ and returns $y_i=f_i(x_1,\dots,x_n)$ to $A_i$.
\end{comment}

\medskip

The value of each point $y \in \mathcal{U}$ is given by the square loss function: $v(y)= -\left| y- \overline{a} \right|^2$; for completeness, we define $v(\bot)=-\infty$.
\footnote{More generally, we could define $v(y)= -d(y,\overline{a})^t$ as the $t$-th moment of the metric distance between $y$ and $\overline{a}$, where $d(\cdot)$ denotes the metric equipped by $U$.},\footnote{The value function $v(\cdot)$ depends on $\overline{a}$ computed from true types $(a_1,\dots, a_n)$ rather than the reported  $(x_1,\dots, x_n)$. This is meaningful, for instance, in applications in which the principal has already collected the types of the agents, and the only choice left to the agents is whether to allow their type to be used in the computation or not.}
The information benefit of a player $i$ on receiving the value $y_i$ from a mechanism $v_i = v(y_i)-v(a_i)$ if $x_i=a_i$ and $v_i=0$ otherwise.
The utility of player $i$ is the same as before: $u_i = v_i - \max_{j \neq i} v_j$.
%We consider the same envy functions $Envy(v_{-i}) = \lambda \max_{j\ne i}v_j$ and $Envy(v_{-i}) = \lambda  \sum_{j\ne i}v_j$ as before,
%Recalling that the information gain of $A_i$ is $v_i = v(y_i)-v(a_i)$ if $x_i=a_i$ and $v_i=0$ otherwise,
%and the utility of $A_i$ is $u_i(a_i,y_i; a_{-i},y_{-i})=v_i-Envy(v_{-i})$.

%%However, we only consider the all-or-nothing model where $\strategy_i=\set{a_i,\bot}$ so that $g_i=y_i$ if $x_i=a_i$ and $g_i=a_i$ if $x_i=\bot$.
%\cnote{remove this paragraph since we never consider flexible?} We remark that we are only concerned about the all-or-nothing model for this problem.
%This is because the problem computes a \emph{reversible} function and hence the flexible model makes no sense. 
%Here by reversible we mean that an agent $A_i$ could infer the value of $y_i$, which is the result he would get if he had submitted $x_i$, by submitting another $\widehat{x}_i$ and receiving $\widehat{y}_i$.
%Thus in the flexible model  $A_i$ does not suffer a punishment when deviating from truthful behavior.
%For example, $A_i$ could claim $x_i$ which is far away from $a_i$ but remove the error in $y_i$ (or at least use $a_i$ instead),
%while $T$ is severely mislead since he cannot tell if $A_i$ is reporting truthfully.
%

\medskip

Our main result in this section is a mechanism for this problem.

\begin{theorem}
	There is a truthful polynomial time mechanism for the average point problem among any number $n$ of all-or-nothing players.
\end{theorem}
\begin{proof}
Let $(x_1,\dots, x_n)$ be the reported inputs, we now design the mechanism by specifying $y_i$'s that $T$ assigns to the agents.
Since all participants appear in equal positions, 
the mechanism is defined as follows:
\begin{align}
\label{eq: avg point}
y_i=f_i(x_1,x_2,\dots,x_n)=
\begin{cases}
\overline{x} = \frac{ \sum_{x_i\ne \bot} x_i }{\text{$\#$ of $x_i$'s such that $x_i\ne \bot$}} & \text{if } x_i\ne \bot\\
\bot		  & \text{if } x_i= \bot
\end{cases}
\end{align}
That is, the mechanism computes the average point $\overline{x}$ of all reported points,
 and sends $\overline{x}$ to all participants (and $\bot$ to the nonparticipants). 

The above mechanism $f$ trivially satisfies the properties of ``maximal social welfare with Pareto efficiency", ``symmetry" and ``null agent gets zero", since all participants get identically the optimal result that the mechanism could offer. 
\footnote{Note the strong dominance property does not make much sense for the average point problem, 
because in each execution there are merely two possible outputs by (\ref{eq: avg point}) and the only inferior relation in results reflects participation/non-participation.}

We now have to argue that $f$ also guarantees the players submit their data points.

Consider the case where a subset $S\subseteq N$ of agents chooses to participate and others do not.
Then, for every player $i \in S$, we have $y_i=\overline{x}=\sum_{A_i\in S} x_i /  \left|S \right|$.
If an agent, say $1 \in S$, deviates by switching from $x_1=a_1$ to $\widehat{x}_1=\bot$, then their result changes from $y_1=\overline{x}$ to $\widehat{y}_1=\bot$ and thus their information benefit changes from $v_1=v(\overline{x})-v(a_1)$
%=\left|a_1-\overline{a}\right|^2-\left|\overline{x}-\overline{a}\right|^2$ 
into $\widehat{v}_1=0$.
That is,
$$\Delta v_1= \widehat{v}_1-v_1= v(a_1)-v(\overline{x})$$
%$\left|\overline{x}-\overline{a}\right|^2- \left|a_1-\overline{a}\right|^2$. 

For every other agent $ i \in S\bs\set{1}$ the result changes from $y_i=\overline{x}$ to
$$\widehat{y}_i=\overline{x}'=\sum_{i\in S\bs\set{1}} x_i /  \left(\left|S \right|-1\right),$$
and so $\Delta v_i = \widehat{v}_i-v_i=v(\overline{x}')-v(\overline{x})$ for every $i=2,3,\dots, n$.
%= - \left|\overline{x}'-\overline{a}\right|^2 + \left|\overline{x}-\overline{a}\right|^2$.
It suffices to show that $\Delta v_1\le \Delta v_i$ for all $i>1$. 
%For $\lambda=1$ we need to show $v(a_1)\le v(\overline{x}')$ (or $v(a_1)-v(\overline{x})\le (|S|-1)\left(v(\overline{x}')-v(\overline{x})\right)$ respectively).

When all agents participate in the computation and $S= N $,
we have $\overline{x}=\overline{a}$ and hence $v(\overline{x})=0$, $\Delta v_i=v(\overline{x}')\le 0$.
Then, by deviating from cooperation player $1$ changes the information benefit of others from $y_i=\overline{x}$ to $\widehat{y}_i=\overline{x}'= \left(n\overline{a}-a_1\right)/(n-1)$.
The value of $\widehat{y}_i$ is as follows:
$$v(\widehat{y}_i)= v(\overline{x}')=- \left|\overline{x}'-\overline{a}\right|^2=-\left| \frac{\overline{a}-a_1}{n-1}  \right|^2\ge 
\frac{-\left|  a_1-\overline{a} \right|^2}{n-1}=\frac{1}{n-1} v(a_1)$$
%$\left|\overline{x}'-\overline{a}\right|^2 \le \left|a_1 - \overline{a}\right|^2$.

Therefore, recalling that $v(\overline{x})=0$, $\Delta v_i=v(\overline{x}')-v(\overline{x})=v(\overline{x}')\le 0$  for $i=2,3,\dots, n$,
\begin{align*}
\Delta v_1 = & v(a_1)-v(\overline{x})
\le (n-1) v(\overline{x}') -v(\overline{x}) 
= (n-1)\left( v(\overline{x}') -v(\overline{x})  \right)
= (n-1)\Delta v_i 
=  \sum_{i=2}^n \Delta v_i
\end{align*}
%\cnote{Maybe I am tired but I don't follow line 3 above, nor the conclusion?}

Thus the players will act truthfully 
since the utility of every agent will decrease when deviating from telling the truth.

We further remark that the above argument holds not only for the square loss function $v(y)=-\left|y-\overline{a}\right|^2$, but also naturally extends to  $v(y) = -\left|y -\overline{a}\right|^p$ when $p>0$.
\end{proof}

\newpage

\section{General Sharing Problems} \label{app:general}

In this section, we investigate the more general case when the value function $v$ satisfies only minimal assumptions. 
More specifically, we introduce a subgroup value function $V:\zo^{|N|}\to \Re$ such that
for every subset $S\subseteq N$, the collection value $V(S)$ 
is defined as the value function $v(\cdot)$ evaluated on the optimal combination of all private types of agents in $S$.
Moreover, the function $V$ is monotone in the sense that for every $S'\subseteq S\subseteq N $, $V(S')\le V(S)$.

The utility of an agent is defined as before, i.e. the difference between the information gain of $i$ and the maximum gain of any other player.
In the following, we will say that a set of players cooperate if they submit their true input.

\begin{definition}
Let $S$ be the set of cooperative agents and $V$ be the subgroup value function.
The \emph{rewardable contribution} of $i\in S$ in the coalition $S$ is denoted by $\phi_i(S)$ such that
\begin{align}
\phi_i(S)
=\max_{T\subseteq S, i\in T}\Big\{
								\min\big\{ V(T)-V(\{i\}), V(T)-V(T\backslash \{i\})  
									\big\}  
							\Big\}
\label{def:phi_S}
\end{align}
\end{definition}
Intuitively, Mechanism~\ref{mech: general} rewards every agent with respect to his contribution within the feasible amount.

\bigskip

\begin{algorithm}[H]\label{mech: general}
\caption{General Mechanism}

	\SetAlgoLined
	\KwIn{  $x_i$ for each player $i$.  }
	\KwOut{ $y=\left(y_1,\dots, y_n \right)$ with $y_i$ assigned to player $i$. }
	
	\medskip
	%pick a random traversal function $\pi$ over $U$
	
	$S = \set{i\st x_i\ne \bot}$
	
	\ForEach{$i\notin S$} {$y_i = \bot$}
	
	\ForEach{$i\in S$}{
		find $y_i\in\result$ such that $v(y_i)=\phi_i(S)$ 
		
		\emph{[in case such $y_i$ does not exist, select $y_i$ such that $\exval{v(y_i)}=\phi_i(S)$]}
		
	}
%	 
%	\If{$j\ne k$}{
%	select $\beta_j',\alpha_k'$ such that $v(y_j)-v(x_j)=v(y_k)-v(x_k)>0$;
%	
%	 $y_j \gets [\alpha_j,\beta_j']$ and $y_k\gets [\alpha'_k,\beta_k]$;
%	 
%	}

	%\KwOut{ $y=\left(y_1,\dots, y_n \right)$}
	Output $\left(y_1,\dots, y_n \right)$.

\end{algorithm}
\bigskip

\begin{theorem}\label{thm:general}
Mechanism~\ref{mech: general} is truthful for any number $n$ of all-or-nothing players.
It is symmetric and satisfies strong dominance given a universal traversal function over the type set of the players.
\end{theorem}

\begin{corollary}
	Mechanism~\ref{mech: general} can be modified to a Pareto optimal mechanism  for $v = \max_i \phi_i(S)$ by Lemma~\ref{lemma:PO condition}.
\end{corollary}

\begin{proof}[Proof of Theorem~\ref{thm:general}]
For every agent $i\notin S$, it holds that $x_i=\bot,y_i=\bot$ and hence $v_i=0$.
Thus Mechanism~\ref{mech: general} makes sure that non-participating agents gets zero information benefit.

The symmetry property is also clear since equivalent agents would have the same rewardable contribution.

Then, we prove the stability of the cooperation. 
For $i\in S$,  we have $x_i=a_i$ and $v(y_i)=\phi_i(S)$, where $\phi_i(S)$ is defined as in (\ref{def:phi_S}).
%
%
%
%%\footnote{The cooperation remains stable if we remove the maximal over $T\subseteq S$ in (\ref{def:phi_S}), but the resulted  social welfare is strictly weaker.}
%
%
Such $y_i$ exists since:
(a) $v(y_i)\le V(S)-v(a_i)$, i.e.~there is sufficient information in the synthesis result of $S$ to reward $i$ with $\phi_i(S)$;
(b) let $\exval{v(y_i)}=\phi_i(S)$ for a randomized $y_i$ in case the range $v(\result)$ is discrete (perhaps because $\result$ is a discrete set or $v(\cdot)$ is segmented).
Intuitively, $\phi_i(S)$ refers to the maximal contribution of $i$ in every possible subset $T\subseteq S$,
and it is also bounded by the maximal available reward within $T$.

In the following, we prove that cooperating is a weakly dominant strategy for every single agent $i$.
Note that the utility of $i$ is $u_i =\phi_i(S)-\max_{j\in S\backslash \{i\}}\big\{ \phi_j(S)\big\}$ for $i\in S$ and it would be $\widehat{u}_i= -\max_{j\in S\backslash \{i\}}\big\{ \phi_j(S\backslash \{i\})\big\}$ if $i$ does not participate.
Then, $i$  will stay in the coalition as long as the following inequality holds:
 $$u_i=\phi_i(S)-\max_{j\in S\backslash \{i\}}\big\{ \phi_j(S)\big\}\ge -\max_{j\in S\backslash \{i\}}\big\{ \phi_j(S\backslash \{i\})\big\}=\widehat{u}_i$$

It suffices to show that for all other agent $j$ where $j\neq i$, 
\begin{align}
\phi_i(S)\ge  \phi_j(S) - \phi_j(S\backslash \{i\})
\label{ineq:Ai}
\end{align}
By the definition of $\phi_i,\phi_j$ as in (\ref{def:phi_S}), the above inequality (\ref{ineq:Ai}) transforms into 
\begin{small}
\begin{align*}
&\max_{T_1\subseteq S, i\in T_1}\Big\{
								\min\big\{ V(T_1)-V(\{i\}), V(T_1)-V(T_1\backslash \{i\})  
									\big\}  
							\Big\}\\
& \ge \max_{T_2\subseteq S, j\in T_2}\Big\{
								\min\big\{ V(T_2)-V(\{j\}), V(T_2)-V(T_2\backslash \{j\} ) 
									\big\}  
							\Big\}\\
& -\max_{T_3\subseteq S\backslash\{i\}, j\in T_3}\Big\{
								\min\big\{ V(T_3)-V(\{j\}), V(T_3)-V(T_3\backslash \{j\}  )
									\big\}  
							\Big\}
\end{align*}
\end{small}

Let $T'$ be the selected set $T_2$ that achieves the maximal in $\phi_j(S)$.
If $i\notin T'$, then the above inequality holds trivially since the right hand side is $0$ while the left hand side non-negative.
Now we focus on the case $i\in T'$, where it suffices to show that 
\begin{small}
\begin{align*}
\min\big\{ V(T')-V(\{i\}), V(T')-V(T'\backslash \{i\}  )
		\big\} 
\ge & \min\big\{ V(T')-V(\{j\}), V(T')-V(T'\backslash \{j\}  ) 
		\big\} -  \\
& - \min\big\{ V(T'\backslash\{i\})-V(\{j\}), V(T'\backslash\{i\})-V(T'\backslash \{i,j\} ) 
		\big\} 
%		\\
%\Leftarrow &\\
% & V(T)- \max\big\{V(\{i\}), V(T\backslash \{i\})  \big\} \\
%\ge & V(T)- \max\big\{V(\{j\}), V(T\backslash \{j\} ) \big\}  \\
%- &V(T\backslash\{i\})+ \max\big\{V(\{j\}),V(T\backslash \{i,j\})  \big\} \\
%\Leftarrow &\\
% & \max\big\{V(\{j\}), V(T\backslash \{j\})  \big\}   + V(T\backslash\{i\})\\
%\ge &\max\big\{V(\{i\}), V(T\backslash \{i\})  \big\} + \max\big\{V(\{j\}),V(T\backslash \{i,j\})  \big\} \\
\end{align*}
\end{small}

Recall that 
$\min\big\{ \alpha-\beta, \alpha-\gamma \big\} =\alpha-\max\set{ \beta, \gamma} $. Then 
the previous inequality can be simplified to
\begin{small}
\begin{align*}
V(T')- \max\big\{V(\{i\}), V(T'\backslash \{i\})  \big\} 
\ge & V(T')- \max\big\{V(\{j\}), V(T'\backslash \{j\} ) \big\} - \\
& - V(T'\backslash\{i\})+ \max\big\{V(\{j\}),V(T'\backslash \{i,j\})  \big\} \\
\end{align*}
\end{small}
and eventually it transforms to the inequality:
\begin{small}
\begin{align}\label{ineq:last}
\max\big\{V(\{j\}), V(T'\backslash \{j\})  \big\}   + V(T'\backslash\{i\}) 
\ge &\max\big\{V(\{i\}), V(T'\backslash \{i\})  \big\} + \max\big\{V(\{j\}),V(T'\backslash \{i,j\})  \big\} 
\end{align}
\end{small}

We prove inequality (\ref{ineq:last}) by case analysis:
\medskip

\begin{itemize}
\item \emph{Case 1:} $V(\set{j}) \ge V(T'\backslash \set{j} )$. 
Then note that 
\begin{small}
\[
V(T'\backslash \set{i} )\ge V(j)\ge V(T'\backslash \set{j} )\ge V(\set i) 
\]
\end{small}
and 
\begin{small}
\[
V(\set j) \ge V(T'\backslash \set j)\ge V(T'\backslash \{i,j\})
\]
\end{small}
Thus, the left hand side is 
\begin{small}
\[
LHS=V(\set{j})+V(T'\backslash \set{i})= V(T'\backslash \set{i})+V(\set{j} )=RHS.
\]
\end{small}

\item \emph{Case 2:} $V(\set{j}) < V(T'\backslash \set{j} )$. 
Then the left hand side is $LHS=V(T'\backslash \set{j})+V(T'\backslash \set{i})$.
Notice that the maximal of $V(T'\backslash \set{j})$ and $V(T'\backslash \set{i})$ satisfies $$\max\{V(T'\backslash \set{j}), V(T'\backslash \set{i})\}\ge \max\big\{V(\set{i}), V(T'\backslash \set{i} )  \big\},$$
while the other satisfies the inequality 
$$\min\{V(T'\backslash\set j),V(T'\backslash \set i)\}\ge \max\{ V(\set j), V(T'\backslash \{i,j\} \}$$
Therefore we have the left hand side (LHS) is greater than or equal to the right hand side (RHS).
\end{itemize}

Thus, we have proved inequality (\ref{ineq:last}) which implies (\ref{ineq:Ai}), and hence $u_i\ge \widehat{u}_i$.
As a result, cooperating is the rational choice for every agent $i$ in the single-envy and all-or-nothing model, which completes the proof.
\end{proof}

We remark that in this mechanism the most beneficial agent gets a reward bounded by his own contribution (i.e. $\max_{T\subseteq S}\set{V(T)-V(T\bs\set{i})}$) and hence not necessarily the maximal feasible (i.e. $V(S)-V(\set{i})$).

\end{document}